\documentclass[letterpaper, 10pt, conference]{ieeeconf}
\IEEEoverridecommandlockouts 

\usepackage{amsthm,amsmath,amssymb,amsfonts}
\usepackage{graphicx}
\usepackage{algorithm}
\usepackage{algpseudocode} 
\usepackage{cancel}
\usepackage{hyperref}
\usepackage{textcomp}
\usepackage[usenames,dvipsnames]{xcolor}
\usepackage{nicefrac} 
\usepackage{dsfont}
\usepackage{tikz} 
\usepackage{pgfplots}
\usepackage{pifont}
\usepackage{cleveref}
\usepackage[backend=bibtex, style=ieee]{biblatex} 
\usepackage{emak} 
\usepackage{setspace}
\usepackage{subcaption}

\begin{document}

\title{A Linear Push-Pull Average Consensus Algorithm for Delay-Prone Networks}

\author{Evagoras~Makridis and Themistoklis~Charalambous
\thanks{The authors are with the Department of Electrical and Computer Engineering, School of Engineering, University of Cyprus, 1678 Nicosia, Cyprus. E-mails: \texttt{\footnotesize surname.name@ucy.ac.cy}. T. Charalambous is also a Visiting Professor at the Department of Electrical Engineering and Automation, School of Electrical Engineering, Aalto University, 02150 Espoo, Finland.}
\thanks{This work was partly supported by the European Research Council (ERC) Consolidator Grant MINERVA (Grant agreement No. 101044629).}
}

\maketitle
\begin{abstract}
In this paper, we address the average consensus problem of multi-agent systems for possibly unbalanced and delay-prone networks with directional information flow. We propose a linear distributed algorithm (referred to as RPPAC) that handles asynchronous updates and time-varying heterogeneous information delays. Our proposed distributed algorithm utilizes a surplus-consensus mechanism and information regarding the number of incoming and outgoing links to guarantee state averaging, despite the imbalanced and delayed information flow in directional networks. The convergence of the RPPAC algorithm is examined using key properties of the backward product of time-varying matrices that correspond to different snapshots of the directional augmented network.
\end{abstract}


\section{Introduction}
Distributed consensus algorithms have recently gained significant prominence due to their widespread applicability in various domains, including wireless sensor networks, multi-agent systems, and smart power grids. The main objective of these algorithms is to enforce a network of interconnected agents (or nodes) to collectively converge towards a common value, known as the \emph{consensus}, by means of local information exchange (see \cite{olfati2007consensus} for an overview of consensus methods). Among the diverse consensus problems, \emph{average consensus} is distinguished as a fundamental challenge, wherein the agents aim to cooperatively compute the average of their initial values held by each individual node. This problem has gained significant attention for its relevance in tasks such as distributed formation control \cite{ren2007multi,fax2004information,yang2023distributed}, distributed estimation and filtering \cite{wang2017diffusion,talebi219distributed,ren2017distributed,lian2020distributed}, and distributed optimization \cite{xin2018linear,pu2020push,nedic2023ab,wang2022surplus,10273441}.

Substantial research on distributed average consensus has been conducted in the context of bidirectional networks where the information between agents flow in both directions forming an undirected graph \cite{xiao2004fast,esteki2022fastest,sebastian2023accelerated}. However, in most real-world multi-agent systems, agents communicate over wireless channels in which they have different transmitting power capabilities and experience different levels of interference. As a consequence the communication links become inherently directional forming directed graphs (\emph{digraphs}). Although directional networks provide a more realistic characterization of the underlying communication network, further complexities and challenges arise when agents aim to achieve average consensus in a distributed manner. For instance, the imbalance in the flow of information, due to the directed links, may lead to slower convergence or even incapability in reaching average consensus. 

To overcome this issue, the works in \cite{dominguez2010coordination, franceschelli2010distributed, cai2012average} have provided different approaches where agents update their state variables through a linear combination of the incoming (received and owned) information. A necessary condition towards reaching average consensus is to prevent information from becoming trapped at a specific agent and to guarantee that the information flows throughout the network, thus, the underlying digraph should be strongly connected\footnote{A digraph is called strongly connected if there exists a directed path between any pair of nodes in the network. This ensures that the information propagates and reaches all agents in the network.}. The \emph{Ratio Consensus} (RC) algorithm proposed in \cite{dominguez2011distributed}, proved to be able to reach average consensus by computing in an iterative way the ratio of two concurrently running linear iterations: one to compute a weighted average of the network's initial values, and one to track the information flow imbalance. Nevertheless, such methods that require agents to send their state variables by weighting them according to the number of their outgoing links (forming a column-stochastic (CS) matrix of weights), require additional computation per iteration for each agent, that is the nonlinear computation of the ratio of the two concurrently running iterations. Another approach towards mitigating the information flow imbalance of digraphs for reaching average consensus has been proposed in \cite{cai2012average}. In particular, the authors proposed a surplus-based approach where the information flow imbalance is handled by augmenting the state of each agent with an additional variable (often called \emph{surplus}) that locally tracks individual state updates through the assignment of weights on both the incoming and outgoing links that form both row-stochastic (RS) and CS matrices, respectively.

Taking one step further towards more realistic scenarios, one should take into consideration that the information flows over delay-prone communication links due to network congestion and packet retransmissions requested as a result of decoding errors detected at the receiving agent \cite{xiao2008consensus,hadjicostis2013average,zhao2017performance,makridis2023utilizing,makridis2023harnessing}. This highlights the necessity of ensuring resilience to delays towards reaching average consensus. Under the conditions on delay-prone directed information exchange between agents, the authors in \cite{hadjicostis2013average} proposed a robustified version of the ratio consensus algorithm (hereinafter referred to as RRC) which is able to reach asymptotic average consensus in the presence of bounded time-varying delays. However, the nonlinear nature of RRC makes it difficult to find tight bounds on the convergence rate of the algorithm itself \cite{hadjicostis2013average}, but also to various consensus-based distributed optimization algorithms that use RC as its consensus protocols \cite{lin2022subgradient,xi2017dextra,nedic2017achieving,xi2017add,xi2017distributed}.

This paper aims at developing a linear discrete-time distributed average consensus algorithm which operates over directed networks, and can handle heterogeneous and time-varying information delays. Specifically, we propose a linear surplus-based distributed protocol that enforces nodes in the network to converge to the average of their initial values, although the communication links are directional and prone to time-varying delays. To the best of our knowledge, this is the first linear (push-pull) asynchronous average consensus algorithm that reaches the exact average of the agents' initial values over delay-prone directional links in a network. The characteristics of our proposed algorithm in comparison with main average consensus algorithms that operate in digraphs, are emphasized in Table~\ref{tab:comparison}.

\begin{table}[h]
    \small
    \centering
    \renewcommand{\arraystretch}{1.3}
    \begin{tabular}{c||c|c|c|c}
    \hline \text { Algorithm } & \text { Linear } & \text { RS } & \text {CS } & \text { Delays } \\
    \hline \hline 
    \text { \cite{dominguez2011distributed} } & \xmark & \xmark & \cmark & \xmark \\\hline 
    \text { \cite{hadjicostis2013average} } & \xmark & \xmark & \cmark & \cmark \\\hline 
    \text { \cite{cai2012average} } & \cmark & \cmark & \cmark & \xmark \\\hline 
    \text { \textbf{RPPAC} } & \cmark &  \cmark & \cmark & \cmark \\
    \hline
    \end{tabular}
    \caption{Comparison of the characteristics of the main average consensus algorithms in digraphs.}
    \label{tab:comparison}
\end{table}

\section{Preliminaries}\label{sec:background}
\subsection{Network Model}
Consider a group of $n>1$ agents communicating over an unreliable time-invariant and directed network. The interconnection topology of the communication network is modeled by a digraph $\set{G}=(\set{V}, \set{E})$. Each agent $v_j$ is included in the set of digraph nodes $\set{V}=\{v_1, \cdots, v_n\}$. The interactions between agents are included in the set of digraph edges $\set{E} \subseteq \set{V} \times \set{V}$. The total number of edges in the network is denoted by $m=|\set{E}|$. A directed edge $\varepsilon_{ji} \triangleq (v_j, v_i) \in \set{E}$ indicates that node $v_j$ receives information from node $v_i$, \ie $v_i \rightarrow v_j$. The nodes that transmit information to node $v_j$ directly are called in-neighbors of node $v_j$, and belong to the set $\inneighbor{j}=\{v_i \in \set{V} | \varepsilon_{ji} \in \set{E}\}$. The number of nodes in the in-neighborhood set is called in-degree and is denoted by $\indegree{j} = |\inneighbor{j}|$. The nodes that receive information from node $v_j$ directly are called out-neighbors of node $v_j$, and belong to the set $\outneighbor{j}=\{v_l \in \set{V} | \varepsilon_{lj} \in \set{E}\}$. The number of nodes in the out-neighborhood set is called out-degree and is denoted by $\outdegree{j}= |\outneighbor{j}|$. Each node $v_j \in \set{V}$ has immediate access to its own local state, and thus we assume that the corresponding self-loop is available $\varepsilon_{jj} \in \set{E}$, although it is not included in the nodes' out-neighborhood and in-neighborhood. In $\mathcal{G}$ a node $v_i$ is reachable from a node $v_j$ if there exists a path from $v_j$ to $v_i$ which respects the direction of the edges. The digraph $\mathcal{G}$ is said to be strongly connected if every node is reachable from every other node.

\subsection{Problem Setup}
At each time instant $k\geq0$ each node $v_j \in \mathcal{V}$ maintains a scalar state $x_j(k) \in \mathbb{R}$. For analysis purposes, we define the aggregate state of all nodes by $x(k)=\left( x_1(k), \ldots, x_n(k)\right)^{\top} \in \mathbb{R}^n$. 
The goal of the agents is to collaboratively solve the following discrete-time average consensus (DTAC) problem:
\begin{align}\label{eq:ac_problem}
    \text{Problem 1:} \quad \bar{x} := \dfrac{1}{n} \sum_{i=1}^{n} x_i(0)
\end{align}
where $\bar{x}$ denotes the network-wide average of all agents' initial values, $x_i(0)$. Clearly, in the absence of global knowledge, individual agents are required to execute an iterative distributed algorithm to eventually converge to the initial network-wide average, by updating their states using information received from their neighboring nodes.

\subsection{Average Consensus using Push-Pull Weights (PPAC)}
A linear algorithm for reaching average consensus over directed and strongly connected networks, has been proposed in \cite{cai2012average}. The main idea behind their proposed algorithm is to maintain a time-invariant state sum $\mathbf{1}^T x$, such that the agents do not lose the track of the initial average $x(0)$, which is the main difficulty when the network is asymmetric (modeled via digraph). To achieve this, at each iteration $k$, each agent $v_j \in \set{V}$ maintains a local state variable $x_j(k) \in \mathbb{R}$, and an auxiliary variable $s_j(k) \in \mathbb{R}$ (called \emph{surplus}). The surplus variable locally records the state changes of individual nodes such that $\mathbf{1}^{\top}(x(k)+s(k))=\mathbf{1}^{\top} x(0)$ for all time $k$, where $s(k)=\left( s_1(k), \ldots, s_n(k)\right)^{\top} \in \mathbb{R}^n$. Then, each agent $v_j$ iteratively updates its variables at each time step $k$ as
\begin{subequations}\label{eq:ppac}
\begin{align}
x_j(k+1) & = \gamma s_j(k) + \sum_{v_i \in \inneighbor{j} \cup \{j\}} r_{ji} x_i(k) , \\
s_j(k+1) & = x_j(k) - x_j(k+1) + \sum_{v_i \in \inneighbor{j} \cup \{j\}} c_{ji} s_i(k),
\end{align}
\end{subequations}
initialized at arbitrary $x_{j}(0) \in \mathbb{R}$, and $s_{j}(0)=0$. Each agent $v_j$ assigns the weights for the incoming information based on its in-degree as:
\begin{align}\label{eq:r-weights}
    r_{ji}=\begin{cases}
    \dfrac{1}{1+\indegree{j}},\!\!\!& \text{ if } v_i \in \inneighbor{j} \text{ or } j=i,\\
     0, &\text{ otherwise},
    \end{cases}
\end{align}
where the resulting weights $r_{ji}\geq0$ (often called \emph{``pull" weights}) are the $(j,i)$-th entries of the row-stochastic matrix $R=\{r_{ji}\} \in \mathbb{R}_{+}^{n \times n}$. The assignment of ``pull" weights is straightforward since each agent can easily obtain its in-degree by counting the incoming streams of information. 
Moreover, each agent $v_j$ assigns the weights for the outgoing information based on its out-degree as:
\begin{align}\label{eq:c-weights}
    c_{lj}=\begin{cases}
    \dfrac{1}{1+\outdegree{j}},\!\!\!& \text{ if } v_l \in \outneighbor{j} \text{ or } l=j,\\
     0, &\text{ otherwise}.
    \end{cases}
\end{align}
where the resulting weights $c_{lj}\geq0$ (often called \emph{``push" weights}) are the $(l,j)$-th entries of the column-stochastic matrix $C=\{c_{lj}\} \in \mathbb{R}_{+}^{n \times n}$. Note that, agents are required to have the knowledge of their out-degree to assign the weights $c_{lj}$, hence it is required to have either an estimate of the out-degree \cite{hadjicostis2015robust,charalambous2015distributed} or to compute the out-going streams of information by utilizing 1-bit feedback links \cite{makridis2023utilizing}.

The parameter $\gamma>0$ (often referred to as \emph{surplus gain}) denotes the gain by which the surplus is amplified, such that each agent regulates its convergence speed to the average consensus value. The selection of parameter $\gamma$ requires global knowledge of the network size. 

\begin{rem}
The weight matrices $R$ and $C$ that are formed by assigning the weights as in \eqref{eq:r-weights} and \eqref{eq:c-weights}, preserve row- and column-stochasticity, respectively. This implies that 
$R\mathbf{1} = \mathbf{1}$ and $\mathbf{1}^{\top} C = \mathbf{1}^{\top}$. Moreover, the fact that digraph $\mathcal{G}$ includes self-loops (each agent has access to its own variables), implies that $r_{ii}>0$ and $c_{ii}>0$, $\forall v_i \in \set{V}$. 
\end{rem}


\section{Algorithm Development}\label{sec:algorithm}

In this section we design a linear distributed strategy for each node $v_j \in \set{V}$ in a directed network $\set{G}=(\set{V},\set{E})$ to handle information that is received over delay-prone directional links, such that all the nodes converge to the average consensus value in \eqref{eq:ac_problem}. 
More specifically we assume that, the transmission on the link $\varepsilon_{ji}$ at time step $k$ experiences a delay, $\tau_{ji}(k)$, \ie the (discrete) time interval between the transmitting and receiving time steps. The delays on the transmission links are heterogeneous, time-varying, and bounded, \ie $0 \leq \tau_{ji}(k) \leq \bar{\tau}_{ji} \leq \bar{\tau} < \infty$, where $\bar{\tau}_{ji}$ denotes the maximum delay over the link $\varepsilon_{ji}$, and $\bar{\tau}$ the maximum delay of the network over all the links \ie $\bar{\tau} = \max_{\varepsilon_{ji} \in\set{E}} \{\bar{\tau}_{ji}\}$. Note that, each node $v_j$ has immediate access to its own local state $x_j$, and hence $\tau_{jj}(k)=0$ for all $v_j \in \set{V}$ and all $k\geq0$.

\subsection{Robustified Push-Pull Average Consensus (RPPAC)}
Inspired by the linear distributed algorithm in \cite{cai2012average} for reaching average consensus over reliable directed graphs, we introduce a robustified alternative by which nodes can reach exact average consensus by handling heterogeneous time-varying delays. Although the number of nodes in the network and the interconnecting links are considered fixed, the presence of time-varying and heterogeneous delays affect the way that each node should assign the consensus weights such that the sum of the nodes' values (\ie \emph{mass}) of the network is preserved at each time step. Hence, to preserve the mass of the network fixed, we devise a new linear algorithm (hereinafter called RPPAC) based on the surplus consensus in \cite{cai2012average}, where each node updates its information state (at each iteration) via a linear combination of the (possibly delayed) information state received from its neighbors at that iteration. This algorithm converges to the exact average of the nodes’ initial values, despite the presence of arbitrary, yet bounded time-delays.

In particular, at each iteration $k$, each agent $v_j \in \set{V}$ maintains a local state variable $x_j(k) \in \mathbb{R}$, and an auxiliary surplus variable $s_j(k) \in \mathbb{R}$, that is used to preserve the total mass in the network constant at each time step (\ie $\mathbf{1}^{\top}(\vect{x}(k)+\vect{s}(k)) = \mathbf{1}^{\top} \vect{x}(0),~\forall k\geq0$).  First, each node $v_j$ sets $x_j(0)=V_j$, $s_j(0)=0$, and $0<\gamma<1$. Prior the iterative phase of the RPPAC algorithm, it broadcasts dummy packets and receives an acknowledgment feedback signal from each in-neighbor to acquire its in-degree and out-degree (see \cite{makridis2023utilizing} for more details on the acquisition of out-degree). 
Following, at each iteration $k$, each node $v_j$ performs the following steps:

\noindent\textbf{Broadcasting:}
It broadcasts its own (unweighted) state variable, $x_j(k)$, and a weighted version of its surplus variable, $c_{lj} s_j(k)$, to its out-neighbors $v_l \in \outneighbor{j}$, over possibly delay-prone links $\varepsilon_{lj} \in \set{E}$. The ``push" consensus weights $c_{lj}$ for $l=1,\ldots,n$ can be assigned offline by each node $v_i$, given that the network is fixed, using the ``push" weight assignment strategy in \eqref{eq:c-weights}. This ensures that at each time step $k$ the total mass of the surplus variable is fully (and equally) distributed to the out-neighbors of $v_j$, since $\sum_{v_l \in \outneighbor{j}} c_{lj} = 1$. Notice that, node $v_j$ transmits its local variables without considering that the information sent over its outgoing links $\varepsilon_{lj}$ for any $v_l \in \outneighbor{j}$ might experience a delay. 

\noindent\textbf{Receiving:}
It receives the weighted surplus variables $c_{ji} s_i(k-\delta)$ and the state variables $x_i(k-\delta)$ for all $0 \leq \delta \leq \bar{\tau}_{ji}$ from (possibly some of) its in-neighbors $v_i \in \inneighbor{j}$ that arrived over possibly delay-prone links $\varepsilon_{ji} \in \set{E}$. Upon the arrival of these variables, it scales each received state variable $x_i(k-\delta)$ by the ``pull" consensus weights $r_{ji}(k)$ as $r_{ji}(k)x_i(k)$, where the weight $r_{ji}(k)$ is assigned by each node $v_j$ depending on the possibly delayed packets arrived exactly at time step $k$. The assignment of ``pull" weights is elaborated later in $\S$\ref{sec:convergence_analysis}.


\noindent\textbf{Updating:}
Upon the reception of (possibly delayed) information from the in-neighbors of $v_j$, it updates its own state variable $x_j(k+1)$ and its auxiliary surplus variable $s_j(k+1)$, as follows: 
\begin{subequations}\label{eq:rppac}
\begin{align}
    x_j(k+1) &= \gamma s_j(k) + \!\!\! \underbrace{\sum_{v_i \in \inneighbor{j}\cup \{j\}} \sum_{\delta=0}^{\bar{\tau}_{ji}} r_{ji}(k) x_i(k-\delta) \ell_{ji}(k-\delta)}_{\text{possibly delayed incoming states}},\label{eq:rppac_x}\\
    s_j(k+1) &= g_j(k+1) + \!\!\! \underbrace{\sum_{v_i \in \inneighbor{j}\cup \{j\}} \sum_{\delta=0}^{\bar{\tau}_{ji}} c_{ji} s_i(k-\delta) \ell_{ji}(k-\delta)}_{\text{possibly delayed incoming surplus}},\label{eq:rppac_s}
\end{align}
\end{subequations}
where $g_j(k+1)\triangleq x_j(k)-x_j(k+1)$ and $\ell_{ji}(k-\delta)$ captures the delay on link $\varepsilon_{ji}$ at iteration $k$ as:
\begin{align}
    \ell_{ji}(k-\delta) = \begin{cases}
        1, & \text{ if } \tau_{ji}(k-\delta) = \delta,\\
        0, & \text{ otherwise}.       
    \end{cases}
\end{align}
The parameter $\gamma>0$ is a sufficiently small number depending on the network size and topology, and the length of delays.

\section{Convergence Analysis}\label{sec:convergence_analysis}

In this section, we analyse the convergence of RPPAC, by introducing its vector-matrix augmented form that corresponds to an augmented digraph that models (possibly) time-varying delayed information. To simplify the analysis, we consider that the maximum delay at each link is identical to the maximum delay in the network, \ie $\bar{\tau}_{ji}=\bar{\tau}$. To model all the possible delayed transmissions consider the following augmentation on the original graph $\set{G}$. For each agent $v_j \in \set{V}$, we add $\bar{\tau}$ extra virtual nodes that represent local buffers which propagate the delayed information to its destined agent after $0<\delta\leq\bar{\tau}$ iterations. Hence the total number of nodes in the augmented digraph $\set{G}^{\alpha}=\{\set{V}^{\alpha},\set{E}^{\alpha}\}$ is $\tilde{n} = n(\bar{\tau}+1)$, where the actual agents are indexed by $1, \ldots, n$ and the virtual nodes by $n+1,\ldots,\tilde{n}$. Thus, the virtual nodes $n+1,\ldots,2n$ model the information delayed by $\delta=1$ time step, $2n+1,\ldots,3n$ model the information delayed by $\delta=2$ time steps, and so on. 

Based on the augmented digraph model, we further define the augmented variables $\tilde{\vect{x}}$ and $\tilde{\vect{s}}$ that hold the (possibly delayed) information in the (virtual buffer) nodes as
\begin{align}
\tilde{\vect{x}}(k) &= \big( x^{\top}(k),~x^{(1)}(k),\ldots, x^{(\bar{\tau})}(k) \big)^{\top},\\
\tilde{\vect{s}}(k) &= \big( s^{\top}(k),~s^{(1)}(k),\ldots, s^{(\bar{\tau})}(k) \big)^{\top},    
\end{align}
where $x^{(\delta)}(k)= \big( x^{(\delta)}_1(k), \ldots, x^{(\delta)}_n(k) \big)$, and  $s^{(\delta)}(k)= \big( s^{(\delta)}_1(k), \ldots, s^{(\delta)}_n(k) \big)$. Then we can rewrite the update phase of RPPAC in its vector-matrix form as:
\begin{subequations}\label{eq:rppac_vector_matrix}
\begin{align}
\tilde{\vect{x}}(k+1) & = \tilde{R}(k)\tilde{\vect{x}}(k) + H \tilde{\vect{s}}(k),\\
\tilde{\vect{s}}(k+1) & = J(k)\tilde{\vect{x}}(k) + (\tilde{C}(k)-H)\tilde{\vect{s}}(k),
\end{align}
\end{subequations}
where 
\begin{align}
\small\def\arraystretch{1.2}
\setlength\arraycolsep{1.4pt}
\tilde{R}(k) &\triangleq \left(\begin{array}{cccc}
R^{(0)}(k) & R^{(1)}(k) & \cdots & R^{(\bar{\tau})}(k)\\
I & 0 & \cdots & 0\\
0 & I & \cdots & 0\\
\vdots & \vdots & \ddots& \vdots\\
0 & 0 & \cdots & 0\\
\end{array}\right) \in \mathbb{R}_+^{\tilde{n}\times \tilde{n}},\nonumber\\
H &\triangleq \left(\begin{array}{cccc}
\gamma I & 0 & \cdots & 0 \\
0 & 0 & \cdots & 0\\
\vdots & \vdots & \ddots & \vdots\\
0 & 0 & \cdots & 0\\
\end{array}\right) \in \mathbb{R}_+^{\tilde{n}\times \tilde{n}},\nonumber\\
J(k) &\triangleq \arraycolsep=1.4pt\def\arraystretch{1.2}\left(\begin{array}{cccc}
I-R^{(0)}(k) & -R^{(1)}(k) & \cdots & -R^{(\bar{\tau})}(k) \\
0 & 0 & \cdots & 0\\
\vdots & \vdots & \ddots & \vdots\\
0 & 0 & \cdots & 0\\
\end{array}\right) \in \mathbb{R}^{\tilde{n}\times \tilde{n}},\nonumber\\
\tilde{C}(k) &\triangleq \left(\begin{array}{cccc}
C^{(0)}(k) & I & \cdots & 0 \\
C^{(1)}(k) & 0 & \cdots & 0\\
\vdots & \vdots & \ddots & I\\
C^{(\bar{\tau})}(k) & 0 & \cdots & 0\\
\end{array}\right) \in \mathbb{R}_+^{\tilde{n}\times \tilde{n}}.\label{eq:system_matrices}
\end{align}

The element at the $j$-th row and $i$-th column of $R^{(\delta)}(k) \in \mathbb{R}_{+}^{n \times n}$ and $C^{(\delta)}(k) \in \mathbb{R}_{+}^{n \times n}$, for $\delta=0,1,\ldots,\bar{\tau}$, are determined by:
\begin{align}\label{eq:r_weighted_delayed_link}
r^{(\delta)}_{ji}(k) = \begin{cases}
    \frac{1}{1+|\set{N}_j^{\alpha,\texttt{in}}(k)|}, & \text { if } \tau_{ji}(k-\delta)=\delta, \; \varepsilon_{ji} \in \mathcal{E}, \\
    0,      & \text { otherwise, }\end{cases}
\end{align}
where $|\set{N}_j^{\alpha,\texttt{in}}(k)|$ is the virtual in-degree of the actual node $v_j$ which denotes the number of (possibly delayed) incoming streams of information that arrived at node $v_j$ exactly at time step $k$; and
\begin{align}\label{eq:c_weighted_delayed_link}
c^{(\delta)}_{ji}(k) = \begin{cases}
    \frac{1}{1 + |\set{N}_i^{\texttt{out}|}}, & \text { if } \tau_{ji}(k-\delta)=\delta, \; \varepsilon_{ji} \in \mathcal{E}, \\
    0,      & \text { otherwise. }\end{cases}    
\end{align}

By appending the vectors $\tilde{\vect{x}}$ and $\tilde{\vect{s}}$ of \eqref{eq:rppac_vector_matrix} in a new augmented state, we get the following matrix form representation:
\begin{align}\label{eq:augmented_rppac}
\begin{bmatrix}
    \tilde{\vect{x}}(k+1)\\
    \tilde{\vect{s}}(k+1)
\end{bmatrix}& = M(k) \begin{bmatrix}
    \tilde{\vect{x}}(k)\\
    \tilde{\vect{s}}(k)
\end{bmatrix}.
\end{align}
where
\begin{align}
M(k) \triangleq \begin{bmatrix}
    \tilde{R}(k) & H\\
    J(k) & \tilde{C}(k)-H
\end{bmatrix}
\end{align}

The following theorem states the conditions on which the algorithm in \eqref{eq:augmented_rppac} achieves average consensus. 

\begin{theorem}
    The algorithm in \eqref{eq:augmented_rppac} achieves asymptotic average consensus with the parameter $\gamma>0$ sufficiently small, if and only if the digraph $\set{G}$ is strongly connected, and the transmission delays are bounded, $\tau_{ji}(k) \leq \bar{\tau}_{ji} \leq \bar{\tau}<\infty$ for all $j,i \in \set{V}$.
\end{theorem}

\begin{proof}
A sketch of the proof is provided in the Appendix.    
\end{proof}

\section{Simulation Results}\label{sec:simulation_results}
Consider a delay-prone directed network, $\set{G}$, shown in Fig.~\ref{fig:graph_numerical_example} comprised of ten agents ($n=10$), with each agent $v_j$ executing the RPPAC algorithm as described in \ref{sec:algorithm}. In this example, the agents' initial values are set at their unique identification index, (\ie $x_j(0)=j$ for $j=1, \ldots, n$), while the auxiliary variables are initialized at $s_j(0)=0$ for all $v_j \in \set{V}$. Based on this configuration, the average of the agents' initial values is $\bar{x}=5.5$. At each time step $k$ the packets transmitted over the network's communication links are possibly and uniformly experiencing a time delay $\tau_{ji}\leq \bar{\tau}$ for all $\varepsilon_{ji} \in \set{E}$.

\begin{figure}[H]
    \centering
    \begin{tikzpicture}[>=Latex, thick,shorten >=2pt, shorten <=2pt]
\tikzset{node distance = 1cm and 1cm}

\begin{scope}[every node/.style={font=\footnotesize,circle,thick,draw},
              main/.style={draw=black,thick, draw, circle, fill=black!10, minimum size=0.8cm},]
\node[main] (1) at (5.5,0) {$v_1$}; 
\node[main] (2) at (4,0) {$v_2$}; 
\node[main] (3) at (2.5,0) {$v_3$}; 
\node[main] (4) at (6.5,-1.2) {$v_4$};
\node[main] (5) at (5,-1.2) {$v_5$}; 
\node[main] (6) at (3,-1.2) {$v_6$}; 
\node[main] (7) at (1.5,-1.2) {$v_7$}; 
\node[main] (8) at (5.5,-2.4) {$v_8$}; 
\node[main] (9) at (4,-2.4) {$v_9$}; 
\node[main] (10) at (2.5,-2.4) {$v_{10}$}; 
\end{scope}

\begin{scope}[>={Latex[black]},
              every node/.style={inner sep=0.04cm,font=\scriptsize,fill=white,circle,minimum size=0.1},
              every edge/.style={draw=black, thick, ->,> = latex'}]
\path[->](1) edge[bend left=13] (2);
\path[->](1) edge[bend left=13] (4);
\path[->](2) edge[bend left=13] (1);
\path[->](3) edge[bend left=13] (2);
\path[->](3) edge[bend left=13] (7);
\path[->](4) edge[bend left=13] (5);
\path[->](4) edge[bend left=13] (8);
\path[->](5) edge[bend left=13] (4);
\path[->](5) edge[bend left=13] (6);
\path[->](6) edge[bend left=13] (7);
\path[->](7) edge[bend left=13] (6);
\path[->](7) edge[bend left=13] (3);
\path[->](8) edge[bend left=13] (4);
\path[->](8) edge[bend left=13] (9);
\path[->](9) edge[bend left=13] (10);
\path[->](10) edge[bend left=13] (7);
\path[->](10) edge[bend left=13] (9);

\end{scope}
\end{tikzpicture}
    \caption{Delay-prone digraph $\set{G}$ comprised of 10 agents. Self-loops are allowed but not shown for ease of presentation.}
    \label{fig:graph_numerical_example}
\end{figure}

As a first step towards validating our theoretical results, we run the RPPAC algorithm over different directional networks with heterogeneous time-varying delays bounded by $\bar{\tau}=\{0,2,5\}$, with a fixed surplus gain $\gamma=0.1$. The state variables $x_i$ at each agent $v_i$ converge to the average consensus value $\bar{x}=5.5$, as shown in Fig.~\ref{fig:iterations_x}, while the surplus variables $s_j$, are driven to $0$, as shown in Fig.~\ref{fig:iterations_s}. It is worth mentioning that, the agents successfully converge to the average consensus value, although the surplus gain $\gamma$ is not carefully chosen based on the network topology and size, as well as the length of delays, but it is rather chosen to be relatively small and same for all considered $\bar{\tau}$.

\begin{figure}[H]
     \centering
     \begin{tikzpicture}
  \begin{axis}[
  scaled y ticks=false,
    width=8.5cm,height=3.5cm,
    grid=major,
	tick label style={font=\small},
	label style={font=\small},
	ylabel={state $x_j(k)$},
	xlabel={iteration ($k$)},
    ymin=, ymax=11,
	xmin=-1,xmax=300,
	xtick={0,50,100,150,200,250,300},
    every axis plot/.append style={thick},
    max space between ticks=20,
    legend cell align={left}, 
    legend pos = south east,
    legend columns=1,
    legend style={font=\scriptsize,column sep=1ex},
]

\pgfplotsinvokeforeach{1,...,9} {
   \addplot[color=Dandelion,opacity={0.1*#1},forget plot] table [x=k,y=x#1, col sep=comma]{results/iterations_delay5.txt};
}      
\addplot[color=Dandelion] table [x=k,y=x10, col sep=comma]{results/iterations_delay5.txt};
\addlegendentry{$\bar{\tau}=5$} 

\pgfplotsinvokeforeach{1,...,9} {
   \addplot[color=Maroon,opacity={0.1*#1},forget plot] table [x=k,y=x#1, col sep=comma]{results/iterations_delay2.txt};
}
\addplot[color=Maroon] table [x=k,y=x10, col sep=comma]{results/iterations_delay0.txt};
\addlegendentry{$\bar{\tau}=2$} 

\pgfplotsinvokeforeach{1,...,9} {
    \addplot[color=NavyBlue,opacity={0.1*#1},forget plot] table [x=k,y=x#1, col sep=comma]{results/iterations_delay0.txt};
}     
\addplot[color=NavyBlue] table [x=k,y=x10, col sep=comma]{results/iterations_delay0.txt};
\addlegendentry{$\bar{\tau}=0$} 

\addplot [ultra thick,black,dashed,domain=0:300] coordinates {(0,5.5) (300,5.5)}; \addlegendentry{$\bar{x}=5.5$}

\end{axis}
\end{tikzpicture}
     \vspace{-5pt}
     \caption{State variable $x_j(k)$ at each agent\vspace{-5pt}.}
     \label{fig:iterations_x}
\end{figure}
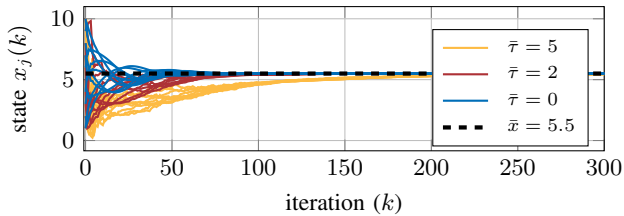

\begin{figure}[H]
     \centering
     \begin{tikzpicture}
  \begin{axis}[
  scaled y ticks=false,
    width=8.5cm,height=3.5cm,
    grid=major,
	tick label style={font=\small},
	label style={font=\small},
	ylabel={surplus $s_j(k)$},
	xlabel={iteration ($k$)},
    ymin=-4, ymax=8,
	xmin=-1,xmax=300,
	xtick={0,50,100,150,200,250,300},
    every axis plot/.append style={thick},
    max space between ticks=20,
    legend cell align={left}, 
    legend pos = north east,
    legend columns=1,
    legend style={font=\scriptsize,column sep=1ex},
]

\pgfplotsinvokeforeach{1,...,9} {
   \addplot[color=Dandelion,opacity={0.1*#1},forget plot] table [x=k,y=s#1, col sep=comma]{results/iterations_delay5.txt};
}      
\addplot[color=Dandelion] table [x=k,y=s10, col sep=comma]{results/iterations_delay5.txt};
\addlegendentry{$\bar{\tau}=5$}  

\pgfplotsinvokeforeach{1,...,9} {
   \addplot[color=Maroon,opacity={0.1*#1},forget plot] table [x=k,y=s#1, col sep=comma]{results/iterations_delay2.txt};
}
\addplot[color=Maroon] table [x=k,y=s10, col sep=comma]{results/iterations_delay0.txt};
\addlegendentry{$\bar{\tau}=2$}

\pgfplotsinvokeforeach{1,...,9} {
    \addplot[color=NavyBlue,opacity={0.1*#1},forget plot] table [x=k,y=s#1, col sep=comma]{results/iterations_delay0.txt};
}     
\addplot[color=NavyBlue] table [x=k,y=s10, col sep=comma]{results/iterations_delay0.txt};
\addlegendentry{$\bar{\tau}=0$} 

\end{axis}
\end{tikzpicture}
    \vspace{-5pt}
     \caption{Surplus variable $s_j(k)$ at each agent\vspace{-5pt}.}
     \label{fig:iterations_s}
\end{figure}
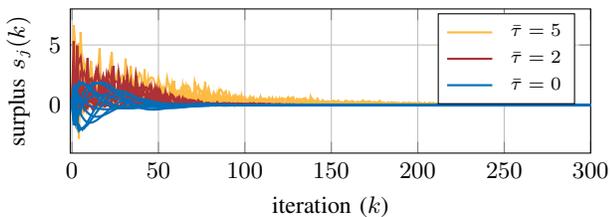

In Fig.~\ref{fig:convergence_rate}, we present the mean consensus error $\frac{1}{n} \vect{e}^{\top}(k) \vect{e}(k)$ where $\vect{e}(k)\triangleq\vect{x}(k) - \mathbf{1} \bar{x}$, achieved by executing the distributed RPPAC algorithm, after averaging over 100 Monte Carlo simulations for three different upper bounds on the delays, \ie $\bar{\tau}=\{0,2,5\}$ and $\gamma=0.1$. The mean square consensus error for the delay-free network is driven to 0 faster than the delay-prone networks, with the same surplus gain parameter $\gamma$.  

\begin{figure}[h]
     \centering
     \begin{tikzpicture}
  \begin{semilogyaxis}[log origin=infty,
  scaled y ticks=false,
    width=8.5cm,height=3.5cm,
    grid=major,
	tick label style={font=\small},
	label style={font=\small},
	ylabel={mean consensus error},
	xlabel={iteration ($k$)},
    ymin=1e-6, ymax=10,
	xmin=0,xmax=300,
	xtick={0,50,100,150,200,250,300},
    every axis plot/.append style={thick},
    max space between ticks=5,
    legend cell align={left}, 
    legend pos = north east,
    legend columns=1,
    legend style={font=\scriptsize,column sep=1ex},
]

\addplot [thick,Dandelion,mark=diamond*,mark repeat=10] table [x=k,y=mean_d5, col sep=comma]{results/convergence_rate_delays.txt};\addlegendentry{$\bar{\tau}=5$}
\addplot [thick,Maroon,mark=triangle*,mark repeat=10] table [x=k,y=mean_d2, col sep=comma]{results/convergence_rate_delays.txt};\addlegendentry{$\bar{\tau}=2$}
\addplot [thick,NavyBlue,mark=square*,mark size=1.4,mark repeat=10] table [x=k,y=mean_d0, col sep=comma]{results/convergence_rate_delays.txt};\addlegendentry{$\bar{\tau}=0$}

\end{semilogyaxis}
\end{tikzpicture}
     \vspace{-15pt}
     \caption{Mean square consensus error for different length of delays\vspace{-5pt}.}
     \label{fig:convergence_rate}
\end{figure}
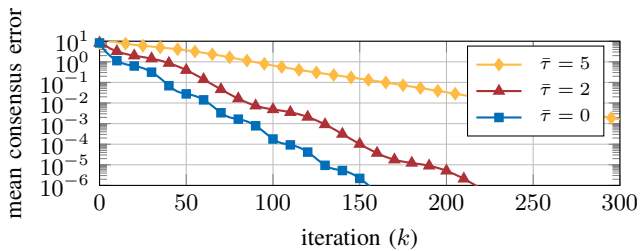
 
In Fig.~\ref{fig:spectral_gap_fixed_maxDelay} we present the spectral gap\footnote{Spectral gap is the difference between the moduli of the two largest eigenvalues of a matrix, \ie $|\lambda_1|-|\lambda_2|$ where $\lambda_i$ is the $i$-th eigenvalue of a matrix $A\in\mathbb{R}^{n \times n}$, with $|\lambda_1|\geq|\lambda_2|\geq\ldots\geq|\lambda_n|.$} of $M(k)$ for different upper bounds on the delays, \ie $\bar{\tau}=\{0,2,5\}$. The higher the spectral gap, the faster the convergence of the RPPAC algorithm, as the second largest absolute eigenvalue of $M(k)$, $\lambda_2$, moves away from the spectral radius $\lambda_1=\rho(M(k))$. As depicted in Fig.~\ref{fig:spectral_gap_fixed_maxDelay}, when the length of delays on the links of the network is longer, the spectral gap of the corresponding matrix $M(k)$ becomes smaller, leading to slower convergence. This behavior is shown in Fig.~\ref{fig:spectral_gap_fixed_gamma} for different upper bounds on the delays, \ie $\bar{\tau}=\{0,1,\ldots,10\}$. However, for a given matrix $M(k)$ that corresponds to a particular snapshot of the interactions in the network, one can choose a surplus gain $\gamma$ that guarantees the fastest convergence to the average consensus value. 

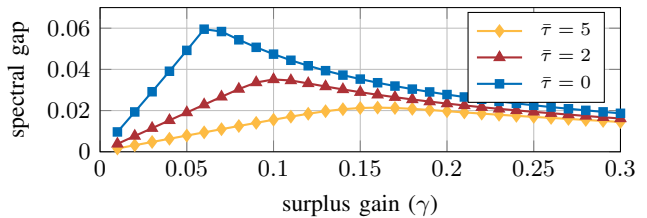
\begin{figure}[h]
     \centering
     \begin{tikzpicture}
  \begin{axis}
  [ ymode=normal,
    xmode=normal,
    scaled ticks=false, 
    tick label style={/pgf/number format/fixed},
    width=8.5cm,height=3.5cm,
    grid=major,
	tick label style={font=\small},
	label style={font=\small},
	ylabel={spectral gap},
	xlabel={surplus gain ($\gamma$)},
    ymin=0, ymax=0.07,
	xmin=0.0,xmax=0.3,
	xtick={0,0.05,...,0.3},
    every axis plot/.append style={thick},
    legend cell align={left}, 
    legend pos = north east,
    legend columns=1,
    legend style={font=\scriptsize,column sep=1ex},
]

\addplot [thick,Dandelion,mark=diamond*] table [x=gamma,y=meanSpectralGap, col sep=comma]{results/fixed_maxDelay5_10nodes.txt};\addlegendentry{$\bar{\tau}=5$}
\addplot [thick,Maroon,mark=triangle*] table [x=gamma,y=meanSpectralGap, col sep=comma]{results/fixed_maxDelay2_10nodes.txt};\addlegendentry{$\bar{\tau}=2$}
\addplot [thick,NavyBlue,mark=square*,mark size=1.4] table [x=gamma,y=meanSpectralGap, col sep=comma]{results/fixed_maxDelay0_10nodes.txt};\addlegendentry{$\bar{\tau}=0$}
    
\end{axis}
\end{tikzpicture}
     \vspace{-15pt}
     \caption{Mean spectral gap of $M(k)$ that corresponds to different for different upper bounds on the delays $\bar{\tau}=\{0,2,5\}$\vspace{-8pt}.}
     \label{fig:spectral_gap_fixed_maxDelay}
\end{figure}

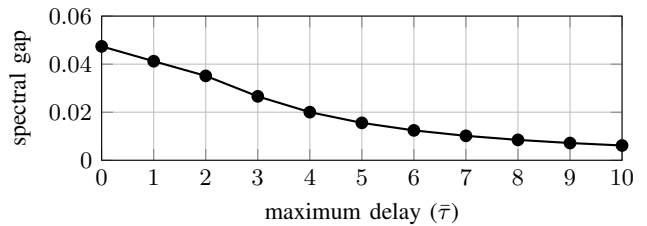
\begin{figure}[h]
     \centering
     \begin{tikzpicture}
  \begin{axis}
  [ ymode=normal,
    xmode=normal,
    scaled ticks=false, 
    tick label style={/pgf/number format/fixed},
    width=8.5cm,height=3.5cm,
    grid=major,
	tick label style={font=\small},
	label style={font=\small},
	ylabel={spectral gap},
	xlabel={maximum delay ($\bar{\tau}$)},
    ymin=0, ymax=0.06,
	xmin=0,xmax=10,
	xtick={0,...,10},
    every axis plot/.append style={thick},
    legend cell align={left}, 
    legend pos = north east,
    legend columns=1,
    legend style={font=\scriptsize,column sep=1ex},
]

\addplot [thick,black,mark=*] table [x=maxDelay,y=meanSpectralGap, col sep=comma]{results/fixed_gamma0.1_10nodes.txt};
    
\end{axis}
\end{tikzpicture}
     \vspace{-15pt}
     \caption{Mean spectral gap of $M(k)$ with $\gamma=0.1$\vspace{-8pt}.}
     \label{fig:spectral_gap_fixed_gamma}
\end{figure}

Following the discussion on the selection of the surplus gain $\gamma$, we run 100 Monte Carlo simulations for the network $\set{G}$ with an upper bound on the delays of $\bar{\tau}=2$, and different surplus gains $\gamma=\{0.01, 0.1, 0.3\}$. The convergence rate in terms of the mean consensus error for this example is shown in Fig.~\ref{fig:convergence_rate_gammas}. Notice that, the fastest convergence using the RPPAC algorithm for this particular configuration, is with $\gamma=0.1$, for which the spectral gap is maximized, as previously shown in Fig.~\ref{fig:spectral_gap_fixed_maxDelay}. 
\begin{figure}[h]
     \centering
     \begin{tikzpicture}
  \begin{semilogyaxis}[log origin=infty,
  scaled y ticks=false,
    width=8.5cm,height=3.5cm,
    grid=major,
	tick label style={font=\small},
	label style={font=\small},
	ylabel={mean consensus error},
	xlabel={iteration ($k$)},
    ymin=1e-6, ymax=10,
	xmin=0,xmax=300,
	xtick={0,50,100,150,200,250,300},
    every axis plot/.append style={thick},
    max space between ticks=5,
    legend cell align={left}, 
    legend pos = north east,
    legend columns=1,
    legend style={font=\scriptsize,column sep=1ex},
]

\addplot [thick,Maroon!40,mark=square*,mark size=1.4,mark repeat=10] table [x=k,y=gamma0.01, col sep=comma]{results/convergence_rate_gammas.txt};\addlegendentry{$\gamma=0.01$}
\addplot [thick,Maroon,mark=triangle*,mark repeat=10] table [x=k,y=gamma0.1, col sep=comma]{results/convergence_rate_gammas.txt};\addlegendentry{$\gamma=0.1$}
\addplot [thick,Maroon!70,mark=diamond*,mark repeat=10] table [x=k,y=gamma0.3, col sep=comma]{results/convergence_rate_gammas.txt};\addlegendentry{$\gamma=0.3$}

\end{semilogyaxis}
\end{tikzpicture}
     \vspace{-20pt}
     \caption{Mean square consensus error with $\bar{\tau}=2$ and $\gamma=\{0.01,0.1,0.3\}$\vspace{-8pt}.}
     \label{fig:convergence_rate_gammas}
\end{figure}
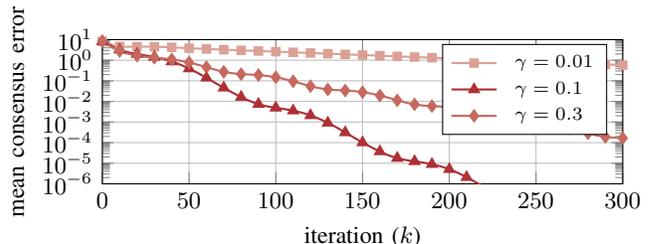

\section{Conclusions and Future Directions}\label{sec:conclusions}

In this paper, we have tackled the discrete-time average consensus in multi-agent systems where the inter-agent communication is directional, potentially unbalanced, and delay-prone. We have introduced a linear distributed algorithm designed to accommodate asynchronous updates and cope with time-varying heterogeneous delays. Our proposed algorithm leverages knowledge about the number of incoming and outgoing links for each agent, thereby ensuring that state averaging is achieved even in the presence of an asymmetric network structure and information flow delays. 


RPPAC brings forward several interesting challenges. For example, what is the optimal choice of $\gamma$ for guaranteeing the fastest convergence? Additionally, it would be interesting to investigate how nodes could choose their own $\gamma$, \ie $\gamma_i$, guaranteeing that convergence is reached.
A promising direction for future research is to study the proposed asynchronous push-pull average consensus algorithm in the context of distributed optimization. 


\appendix
\noindent\textbf{Sketch of the Proof of Theorem 1:}\label{appendix:proofs}\\
We start by partitioning $M(k)$ into the summation of the terms $M_0(k)$ and $M_1$, 
\ie
\begin{align}\label{eq:m0+m1}
    M(k) = M_0(k) + M_1,
\end{align}
where
$$
M_0(k) := \begin{bmatrix}
    \tilde{R}(k) & 0\\
    J(k) & \tilde{C}(k)-H
\end{bmatrix}, \; \text{and } \;
M_1 := \begin{bmatrix}
    0 &  H\\
    0 & 0
\end{bmatrix}.
$$
$M_0 (k)$ is lower block triangular,
its eigenvalues are the union of eigenvalues of its block diagonal elements, that is, the matrices $\tilde{R}(k)$ and $\tilde{C}(k)-H$, \ie its spectrum is 
$$
\sigma(M_0(k)) = \sigma \big(\tilde{R}(k)\big) \; \cup \; \sigma \big( \tilde{C}(k)-H\big),
$$
where $\sigma(\cdot)$ denotes the spectrum, and $\tilde{R}(k)$ and $\tilde{C}(k)$ are the row- and column-stochastic matrices defined in \eqref{eq:system_matrices}, with spectral radius $\rho(\tilde{R}(k))=\rho(\tilde{C}(k))=1$.

Since $\set{G}$ is strongly connected, the corresponding augmented graph $\set{G}^{\alpha}$ that models the delayed information is jointly strongly connected after $\bar{\tau} + 1$ steps \cite{hadjicostis2013average}. 
Therefore, any $\beta-$length word $\bar{E}_{(\beta)}=(\tilde{C}(k+\beta)-H)(\tilde{C}(k+\beta-1)-H)\ldots (\tilde{C}(k+1)-H)$ and $\bar{R}_{(\beta)}=\tilde{R}(k+\beta) \tilde{R}(k+\beta-1) \ldots \tilde{R}(k+1)$, for any integer $\beta\geq\bar{\tau}+1$, gives graphs that are strongly connected.  

We start by defining the augmented state by concatenating the state and surplus variables of all the augmented nodes as $\chi(\cdot) = \begin{bmatrix}
    \tilde{\vect{x}}^{\top}(\cdot),~\tilde{\vect{s}}^{\top}(\cdot)
\end{bmatrix}^{\top}$.
First, consider the state and surplus variables update using the RPPAC in \eqref{eq:augmented_rppac} initialized with $\chi(k_0)$, at iteration $k_2$:
\begin{align}
\chi(k_2) = M(k_1) \chi(k_1)= M(k_1) M(k_0) \chi(k_0).\nonumber
\end{align}
We examine the backward product of matrices $M(k_1) M(k_0)$. From \eqref{eq:m0+m1},
\begin{align}\label{eq:backward-product2}
M&(k_1) M(k_0) =(M_0(k_1)+M_1)(M_0(k_0)+M_1) \nonumber \\
&=M_0(k_1)M_0(k_0) + M_0(k_0) M_1 + M_1 M_0(k_0) + M_1^2\nonumber \\
&\stackrel{(a)}{=} M_0(k_1)M_0(k_0) + M_0(k_0) M_1 + M_1 M_0(k_0), 
\end{align}
where $(a)$ stems from the fact that $M_1^2=\mathbf{0}_{\tilde{n}\times\tilde{n}}$ (comes directly from the definition of $M_1$ in \eqref{eq:m0+m1}).
Next, consider the corresponding backward product of matrices $M(k_2)M(k_1) M(k_0)$, \ie 
\begin{align}\label{eq:backward-product3}
M&(k_2)M(k_1) M(k_0) \nonumber\\
&=(M_0(k_2)+M_1) \big[ M_0(k_1)M_0(k_0) \nonumber\\
& \qquad + M_0(k_0) M_1 + M_1 M_0(k_0) \big]\nonumber\\
&\stackrel{(b)}{=} M_0(k_2)M_0(k_1)M_0(k_0) + M_0(k_2) M_0(k_0) M_1 +\nonumber\\
              & \quad \quad M_0(k_2) M_1 M_0(k_0) + M_1 M_0(k_1)M_0(k_0) +\nonumber\\
              & \quad \quad M_1 M_0(k_0) M_1, 
\end{align}
where $(b)$ stems again from the fact that $M_1^2=\mathbf{0}_{\tilde{n}\times\tilde{n}}$. Continuing in the same way, we can see that a lot of terms that have $M_1$ on the left-side of the product will be cancelled out. 
%
%
Additionally, we can deduce the following properties that hold for any integer $\beta\geq\bar{\tau}+1$, for the $\beta-$length word $\bar{M}_{(\beta)}=M(k+\beta)M(k+\beta-1)\ldots M(k+1)$.
\begin{itemize}
    \item Since the product of lower triangular matrices results in a lower triangular matrix, then the product $\bar{M}_{0,(\beta)} \triangleq M_{0}(k+\beta) M_{0}(k+\beta-1)\ldots M_0(k+1)$ results in a lower-triangular matrix of the form:
        $$
        \bar{M}_{0,(\beta)} = \begin{bmatrix}
            \bar{R}_{(\beta)} & \mathbf{0}\\
            \ast & \bar{E}_{(\beta)}
        \end{bmatrix},
        $$
        where $\rho(\bar{R}_{(\beta)})=1$ which is known from the product of row-stochastic matrices. Then, $\gamma$ should be chosen such that $\rho(\bar{E}_{(\beta)})<1$ \cite{ahmadi2012joint}. This can be secured for $0<\gamma<\underline{c}$, where $\underline{c}\triangleq \min\{C\}$ is the minimum consensus weight of the original column-stochastic matrix $C$ as formed by the weight assignment in \eqref{eq:c-weights}. This comes from the fact that the diagonal elements of the top-left block of dimension $n$ of $\tilde{C}(k)-H$ are strictly positive by enforcing $0<\gamma<\underline{c}$. While this is a conservative bound guaranteeing that $\rho(\tilde{C}(k)-H)<1$~$\forall k$ it gives a simple bound for the values of $\gamma$.
        \item The fact that $M_1^2=\mathbf{0}_{\tilde{n}\times\tilde{n}}$ makes several terms to be cancelled. The remaining products are such that they make the correction for the consensus to reach the average (as demonstrated in the simulations). The proof of this part is tedious and omitted due to space limitation. We will include it in the extended version of the paper.
\end{itemize}

Based on the aforementioned properties of the product of $\beta-$length $M(k)$ matrices, it is guaranteed that the state variables of the networked agents executing the RPPAC algorithm converge to the average consensus value and concurrently their surplus variables are driven to $0$, since the network is jointly strongly connected after $\bar{\tau}+1$ steps.

\begin{rem}
    The convergence rate of the RPPAC algorithm is driven by the convergence of the product $\bar{E}_{(\beta)}$ to $\mathbf{0}_{\tilde{n} \times \tilde{n}}$. In other words, the algorithm achieves average consensus when the surplus variables at each agent $s_i(k)$ have been completely released 
    to the corresponding state variable $x_i(k)$.
\end{rem}


\printbibliography 

@article{olfati2007consensus,
	author = {Olfati-Saber, Reza and Fax, J Alex and Murray, Richard M},
	journal = {Proceedings of the IEEE},
	number = {1},
	pages = {215--233},
	title = {{Consensus and Cooperation in Networked Multi-agent Systems}},
	volume = {95},
	year = {2007}}

@inproceedings{dominguez2010coordination,
	author = {Dominguez-Garcia, Alejandro D and Hadjicostis, Christoforos N},
	booktitle = {IEEE International Conference on Smart Grid Communications},
	pages = {537--542},
	title = {{Coordination and Control of Distributed Energy Resources for Provision of Ancillary Services}},
	year = {2010}}

@inproceedings{dominguez2011distributed,
	author = {Dom{\'\i}nguez-Garc{\'\i}a, Alejandro D and Hadjicostis, Christoforos N},
    booktitle={IEEE Conference on Decision and Control},
	pages = {2124--2129},
	title = {{Distributed Strategies for Average Consensus in Directed Graphs}},
	year = {2011}}

@article{franceschelli2010distributed,
	author = {Franceschelli, Mauro and Giua, Alessandro and Seatzu, Carla},
	journal = {IEEE Sensors Journal},
	number = {3},
	pages = {808--817},
	title = {{Distributed Averaging in Sensor Networks based on Broadcast Gossip Algorithms}},
	volume = {11},
	year = {2010}}

@article{hadjicostis2013average,
	author = {Hadjicostis, Christoforos N and Charalambous, Themistoklis},
    journal={IEEE Transactions on Automatic Control},
	number = {3},
	pages = {763--768},
	title = {{Average Consensus in the Presence of Delays in Directed Graph Topologies}},
	volume = {59},
	year = {2013}}

@article{cai2012average,
  title={{Average Consensus on General Strongly Connected Digraphs}},
  author={Cai, Kai and Ishii, Hideaki},
  journal={Automatica},
  volume={48},
  number={11},
  pages={2750--2761},
  year={2012},
  publisher={Elsevier}
}

@article{hadjicostis2015robust,
  title={{Robust Distributed Average Consensus via Exchange of Running Sums}},
  author={Hadjicostis, Christoforos N and Vaidya, Nitin H and Dom{\'\i}nguez-Garc{\'\i}a, Alejandro D},
  journal={IEEE Transactions on Automatic Control},
  volume={61},
  number={6},
  pages={1492--1507},
  year={2015},
}

@article{charalambous2015distributed,
  title={{Distributed Finite-Time Computation of Digraph Parameters: Left-Eigenvector, Out-Degree and Spectrum}},
  author={Charalambous, Themistoklis and Rabbat, Michael G and Johansson, Mikael and Hadjicostis, Christoforos N},
  journal={IEEE Transactions on Control of Network Systems},
  volume={3},
  number={2},
  pages={137--148},
  year={2015},
}

@inproceedings{makridis2023utilizing,
  title={{Utilizing Feedback Channel Mechanisms for Reaching Average Consensus over Directed Network Topologies}},
  author={Makridis, Evagoras and Charalambous, Themistoklis and Hadjicostis, Christoforos N},
  booktitle={American Control Conference},
  pages={1781--1787},
  year={2023},
  organization={}
}

@inproceedings{makridis2023harnessing,
  title={{Harnessing HARQ Retransmissions for Fast Average Consensus Over Unreliable Communication Channels}},
  author={Makridis, Evagoras and Charalambous, Themistoklis and Hadjicostis, Christoforos N},
  booktitle={IEEE Conference on Decision and Control},
  pages={5437--5443},
  year={2023},
  organization={}
}

@article{xin2018linear,
  title={{A Linear Algorithm for Optimization over Directed Graphs with Geometric Convergence}},
  author={Xin, Ran and Khan, Usman A},
  journal={IEEE Control System Letters},
  volume={2},
  number={3},
  pages={315--320},
  year={2018},
}

@article{pu2020push,
  title={{Push-Pull Gradient Methods for Distributed Optimization in Networks}},
  author={Pu, Shi and Shi, Wei and Xu, Jinming and Nedi{\'c}, Angelia},
  journal={IEEE Transactions on Automatic Control},
  volume={66},
  number={1},
  pages={1--16},
  year={2020},
}

@article{nedic2023ab,
  title={{AB/Push-Pull Method for Distributed Optimization in Time-Varying Directed Networks}},
  author={Nedi{\'c}, Angelia and Nguyen, Duong Thuy Anh and Nguyen, Duong Tung},
  journal={Optimization Methods and Software},
  pages={1--28},
  year={2023},
}

@ARTICLE{talebi219distributed,
  author={Talebi, Sayed Pouria and Werner, Stefan},
  journal={IEEE Transactions on Automatic Control},
  title={{Distributed Kalman Filtering and Control Through Embedded Average Consensus Information Fusion}}, 
  year={2019},
  volume={64},
  number={10},
  pages={4396-4403},}

@article{ren2007multi,
  title={{Multi-Vehicle Consensus with a Time-Varying Reference State}},
  author={Ren, Wei},
  journal={Systems and Control Letters},
  volume={56},
  number={7-8},
  pages={474--483},
  year={2007},
}

@article{fax2004information,
  title={{Information Flow and Cooperative Control of Vehicle Formations}},
  author={Fax, J Alexander and Murray, Richard M},
  journal={IEEE Transactions on Automatic Control},
  volume={49},
  number={9},
  pages={1465--1476},
  year={2004},
}

@article{yang2023distributed,
  title={{Distributed Robust Adaptive Formation Control of Multi-Agent Systems with Heterogeneous Uncertainties and Directed Graphs}},
  author={Yang, Wenlong and Shi, Zongying and Zhong, Yisheng},
  journal={Automatica},
  volume={157},
  pages={111275},
  year={2023},
}

@article{wang2017diffusion,
  title={{Diffusion Distributed Kalman Filter over Sensor Networks without Exchanging Raw Measurements}},
  author={Wang, Guoqing and Li, Ning and Zhang, Yonggang},
  journal={Signal Processing},
  volume={132},
  pages={1--7},
  year={2017},
}

@article{ren2017distributed,
  title={{Distributed Kalman--Bucy Filter with Embedded Dynamic Averaging Algorithm}},
  author={Ren, Wei and Al-Saggaf, Ubaid M},
  journal={IEEE Systems Journal},
  volume={12},
  number={2},
  pages={1722--1730},
  year={2017},
}

@article{lian2020distributed,
  title={{Distributed Kalman Consensus Filter for Estimation with Moving Targets}},
  author={Lian, Bosen and Wan, Yan and Zhang, Ya and Liu, Mushuang and Lewis, Frank L and Chai, Tianyou},
  journal={IEEE Transactions on Cybernetics},
  volume={52},
  number={6},
  pages={5242--5254},
  year={2020},
}

@article{wang2022surplus,
  title={{Surplus-Based Accelerated Algorithms for Distributed Optimization over Directed Networks}},
  author={Wang, Dong and Wang, Zhu and Lian, Jie and Wang, Wei},
  journal={Automatica},
  volume={146},
  pages={110569},
  year={2022},
}

@ARTICLE{10273441,
  author={Khatana, Vivek and Saraswat, Govind and Patel, Sourav and Salapaka, Murti V.},
  journal={IEEE Transactions on Network Science and Engineering}, 
  title={{GradConsensus: Linearly Convergent Algorithm for Reducing Disagreement in Multi-Agent Optimization}}, 
  year={2023},
  volume={},
  number={},
  pages={1-13}}

@article{xiao2004fast,
  title={{Fast Linear Iterations for Distributed Averaging}},
  author={Xiao, Lin and Boyd, Stephen},
  journal={Systems and Control Letters},
  volume={53},
  number={1},
  pages={65--78},
  year={2004},
}

@inproceedings{lin2022subgradient,
  title={{Subgradient-Push is of the Optimal Convergence Rate}},
  author={Lin, Yixuan and Liu, Ji},
  booktitle={IEEE Conference on Decision and Control},
  pages={5849--5856},
  year={2022},
}

@article{xi2017distributed,
	author = {Xi, Chenguang and Wu, Qiong and Khan, Usman A},
	journal = {Neurocomputing},
	pages = {508--515},
	publisher = {Elsevier},
	title = {{On the Distributed Optimization over Directed Networks}},
	volume = {267},
	year = {2017}}

@article{nedic2017achieving,
	author = {Nedic, Angelia and Olshevsky, Alex and Shi, Wei},
	journal = {SIAM Journal on Optimization},
	number = {4},
	pages = {2597--2633},
	publisher = {SIAM},
	title = {{Achieving Geometric Convergence for Distributed Optimization over Time-varying Graphs}},
	volume = {27},
	year = {2017}}

@article{xi2017add,
	author = {Xi, Chenguang and Xin, Ran and Khan, Usman A},
	journal = {IEEE Transactions on Automatic Control},
	number = {5},
	pages = {1329--1339},
	title = {{ADD-OPT: Accelerated Distributed Directed Optimization}},
	volume = {63},
	year = {2017}}

@article{xi2017dextra,
	author = {Xi, Chenguang and Khan, Usman A},
	journal = {IEEE Transactions on Automatic Control},
	number = {10},
	pages = {4980--4993},
	title = {{DEXTRA: A Fast Algorithm for Optimization over Directed Graphs}},
	volume = {62},
	year = {2017}}

@INPROCEEDINGS{ahmadi2012joint,
  author={Ahmadi, Amir Ali and Parrilo, Pablo A.},
  booktitle={IEEE Conference on Decision and Control}, 
  title={{Joint Spectral Radius of Rank One Matrices and the Maximum Cycle Mean Problem}}, 
  year={2012},
  pages={731--733}}

@article{xiao2008consensus,
  title={{Consensus Protocols for Discrete-Time Multi-Agent Systems with Time-Varying Delays}},
  author={Xiao, Feng and Wang, Long},
  journal={Automatica},
  volume={44},
  number={10},
  pages={2577--2582},
  year={2008},
  publisher={Elsevier}
}

@article{zhao2017performance,
  title={{Performance Analysis of Discrete-Time Average Consensus under Uniform Constant Time Delays}},
  author={Zhao, Chengcheng and He, Jianping and Cheng, Peng and Chen, Jiming},
  journal={IFAC-PapersOnLine},
  volume={50},
  number={1},
  pages={11725--11730},
  year={2017},
  publisher={Elsevier}
}

@ARTICLE{sebastian2023accelerated,
  author={Sebastián, Eduardo and Montijano, Eduardo and Sagüés, Carlos and Franceschelli, Mauro and Gasparri, Andrea},
  journal={IEEE Control Systems Letters}, 
  title={{Accelerated Multi-Stage Discrete Time Dynamic Average Consensus}}, 
  year={2023},
  volume={7},
  number={},
  pages={2731-2736},}

@article{esteki2022fastest,
  title={{The Fastest Linearly Converging Discrete-Time Average Consensus using Buffered Information}},
  author={Esteki, Amir-Salar and Moradian, Hossein and Kia, Solmaz S},
  journal={arXiv preprint arXiv:2206.09916},
  year={2022}
}

\end{document}